\documentclass[preprint,12pt]{elsarticle}
\usepackage{bbm}

\usepackage{amssymb}
\usepackage{mathrsfs}
\usepackage{graphicx}
\usepackage{amsmath, amssymb, amsthm}
\usepackage{leftidx}
\usepackage{extarrows}
\usepackage{stfloats}
\usepackage{picinpar}
\usepackage{enumerate}
\usepackage{algorithm}
\usepackage{algorithmicx}
\usepackage{algpseudocode}
\usepackage{epsfig}
\usepackage{latexsym}
\usepackage{amsfonts}
\usepackage{enumerate}
\usepackage{graphics}
\usepackage{MnSymbol}
\usepackage{float}
\usepackage{pict2e}
\usepackage{tikz}
\usepackage{graphicx}
\usepackage{amsmath,amssymb,amsthm}
\usepackage{leftidx}
\usepackage{extarrows}
\usepackage{colortbl}
\usepackage{algorithm}
\usepackage{algorithmicx}
\newtheorem{thm}{Theorem}

\newtheorem{lem}{Lemma}
\newtheorem{rmk}{Remark}
\newtheorem{defn}{Definition}
\newtheorem{example}{Example}

\journal{}

\begin{document}

\begin{frontmatter}



\title{Improved transformation between Fibonacci FSRs and Galois FSRs based on semi-tensor product\tnoteref{label1}}

\author[1]{Bowen Li}\ead{qfhxjy@126.com}
\author[2]{Shiyong Zhu}\ead{zhusy0904@gmail.com}
\author[3]{Jianquan Lu\corref{cor1}}\ead{jqluma@seu.edu.cn}
\address[1]{School of Information Science and Engineering, Southeast University, Nanjing 210096, China}
\address[2]{Department of Systems Science, School of Mathematics, Southeast University, Nanjing 210096, China. 
}
\address[3]{Department of Systems Science, School of Mathematics, Southeast University, Nanjing 210096, China. 
}

\tnotetext[label1]{This work was supported by the National Natural Science Foundation of China under Grant No. 61573102, the Natural Science Foundation of Jiangsu Province of China under Grant BK20170019, the Jiangsu Provincial Key Laboratory of Networked Collective Intelligence under Grant No. BM2017002, Jiangsu Province Six Talent Peaks Project under Grant 2015-ZNDW-002, the Fundamental Research Funds for the Central Universities under Grant No. 2242019k1G013, and Postgraduate Research $\&$ Practice Innovation Program of Jiangsu Province KYCX19$\_$0111.}
\cortext[cor1]{Corresponding author.}
\begin{abstract}
Feedback shift registers (FSRs), which have two configurations: Fibonacci and Galois, are a primitive building block in stream ciphers. In this paper, an improved transformation is proposed between Fibonacci FSRs and Galois FSRs. In the previous results, the number of stages is identical when constructing the equivalent FSRs. In this paper, there is no requirement to keep the number of stages equal for two equivalent FSRs here. More precisely, it is verified that an equivalent Galois FSR with fewer stages cannot be found for a Fibonacci FSR, but the converse is not true. Furthermore, a given Fibonacci FSR with $n$ stages is proved to have a total of $\left(2^{n-1}\right)!^{2}-1$ equivalent Galois FSRs. In order to reduce the propagation time and memory, an effective algorithm is developed to find equivalent Galois FSR and is proved to own minimal operators and stages. Finally, the feasibility of our proposed strategies, to mutually transform Fibonacci FSRs and Galois FSRs, is demonstrated by numerical examples.
\end{abstract}

\begin{keyword}
Feedback shift registers \sep  Boolean networks \sep  semi-tensor product.


\end{keyword}

\end{frontmatter}


\section{Introduction}
Pseudo-random sequences are deterministic sequences with certain random properties, and have a wide range of applications, including but not limited to, detection, encryption, scrambling and spreading. In digital circuits, the common pseudo-random sequences generators are: Feedback shift registers (FSRs), filter generators, combination generators and binary machines. Due to FSRs' conceptual simplicity and effect applications, the investigations on FSRs have attracted considerable attention of researchers. An FSR is composed of a clock, updated functions and $n$ registers, which are also called stages. More precisely, each register, denoted by $x_{i}$, $i=1,2,\cdots,n$, only has the fundamental binary states: $1$ and $0$, where $x_{1}$ and $x_{n}$ are respectively called the lowest and highest order registers. At every clock cycle, the state update of register $x_{i}$ depends on the corresponding update function $f_{i}$, which is composed of some logical operators such as $\vee$ and $\wedge$. The state of one FSR at time instant
$t$, denoted by $X(t)=\left(x_{1}(t), x_{2}(t),\cdots, x_{n}(t)\right)$, is computed by the values of all registers at the last time instant (i.e., $X(t-1)$) and the corresponding update functions $f_i$. In the previous literature, algebraic normal form (ANF) is a common representation for Boolean functions, while addressing the analysis and synthesis of FSRs. In particular, for a Boolean function $f: \{0, 1\}^{n}\rightarrow \{0, 1\}$, its ANF is essentially a polynomial in Galois fields of order $(2)$ (GF$(2)$) as $f(x_{1},\cdots, x_{n})=\sum_{i=1}^{2^{n}}c_{i}\cdot x_{1}^{i_{1}}\cdot x_{2}^{i_{2}}\cdots x_{n}^{i_{n}}$. Thereinto, $c_{i}$ takes value from $\{0, 1\}$, and $(i_{1} i_{2}\cdots i_{n})$ is the binary expansion of $i$ with $i_{1}$ being the least significant bit. Based on the ANF of logical functions, the corresponding polynomial form of FSRs can be further obtained. Thereby, many problems of FSRs were investigated, including but not limited to, irreducibility \cite{Jiang2018TITp3944}, equivalent transformation \cite{Dubrova2009ITITp5263}, decomposition \cite{Zhang2015ITITp645}, as well as attack \cite{Zadeh2014IETp188}.

In general, according to the implementing configurations, FSRs can be divided into two types: the Fibonacci and the Galois, as shown in Fig. \ref{Tf1} \cite{Dubrova2009ITITp5263}. The first one is called a Fibonacci FSR and is conceptually more simple. In Fibonacci FSRs, the registers are chained to each other, then the state of register $x_{i}$, $i\in [2, n]$, is transmitted to the next register $x_{i-1}$, except for register $x_{1}$. The unique feedback function exists to update the state of the $n$-th register, namely $x_n$. While in Galois FSRs, each register has its own feedback function rather than the chain connection form. For both types of FSRs, the value of the lowest register acts as the output of the whole FSRs. Therefore, for a Fibonacci FSR, it is clearly noticed that the period of it state trajectory equals to that of its corresponding output sequence. It is helpful to construct Fibonacci FSRs when the period of output sequences satisfies some special characters. For a Galois FSR, the period of its output sequence is not necessary to be equal to that of states trajectory, but must be a divisor of the period of the corresponding state trajectory. From the application point of view, the depth of circuits used in the updated functions of Galois FSRs is potentially smaller than that of Fibonacci FSRs \cite{Dubrova2009ITITp5263}. Hence, these two kinds of FSRs have own disadvantages and advantages in the practical applications, and it motivates us to study the transformation between Fibonacci FSRs and Galois FSRs.
\begin{figure}[H]
\centering
\includegraphics[width=0.5\textwidth]{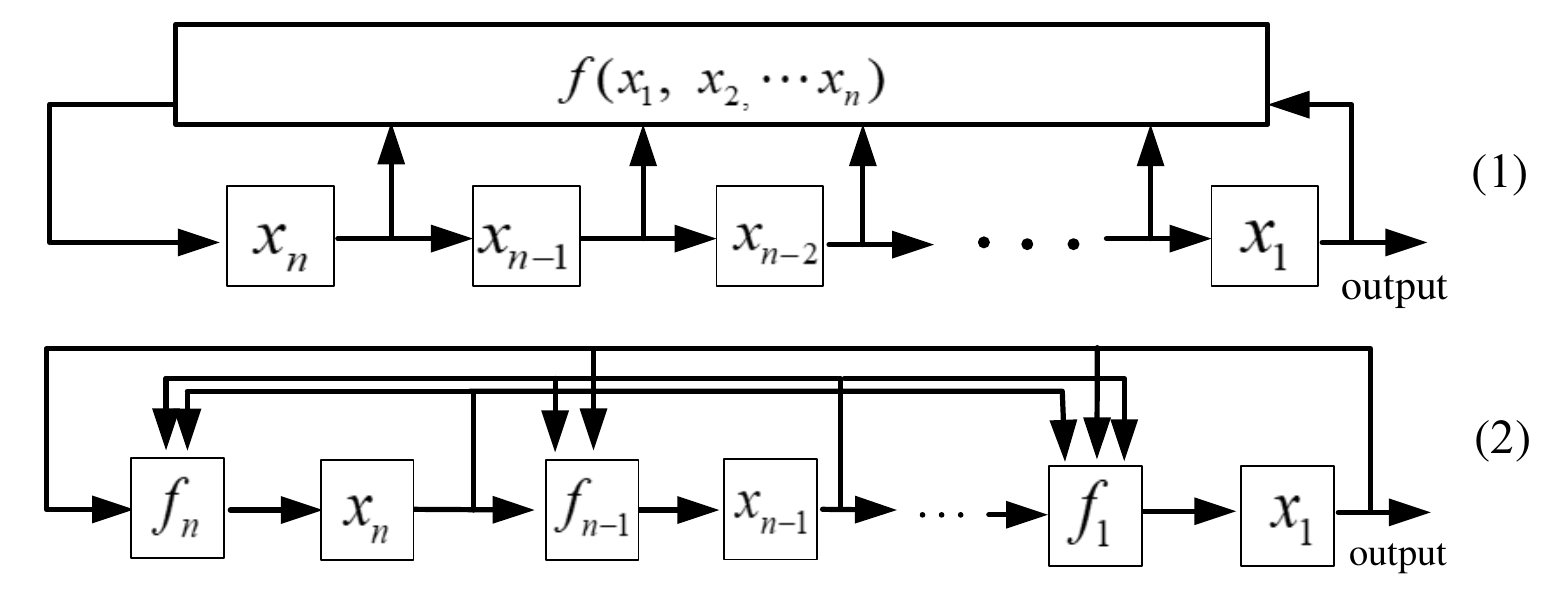}
\caption{The implementing configurations of two types of FSRs: (1) Fibonacci FSRs; (2) Galois FSRs \cite{Dubrova2009ITITp5263}.}
\label{Tf1}
\end{figure}

The transformation between two types of FSRs is to construct a Galois FSR and find the initial state which matches a given initial state of the given Fibonacci FSR to generate same outputs sequences, and vice versa. In \cite{Dubrova2009ITITp5263}, Dubrova proposed a special form called uniform form. Based on this special form and ANF, Dubrova argued that for a given Fibonacci FSR with uniform form, if the Galois FSRs obtained from this Fibonacci FSR are also uniform, then these two FSRs are equivalent. Furthermore, Dubrova in \cite{Dubrova2010ITITp2961} found the matching initial states for the equivalent of these two FSRs with uniform form. Unfortunately, there still exist a large number of FSRs, which are not uniform. Therefore, it is natural to raise such a question: For a Galois FSR which does not satisfy the uniform conditions, whether the corresponding equivalent Fibonacci FSR exists?

To answer the question, some works have been done. In \cite{Liu2013p335}, although the uniform conditions are not satisfied, it requires that the update functions only take input from lower stages than the stages they update. Later on, in \cite{Dubrova2014p187}, Dubrova proposed a novel transformation, which can be applied to arbitrary Galois FSRs rather than just the uniform FSRs, but some constraints on the update functions were still necessary to realize this transformation. Very recently, in order to further relax the constraints, Lu \emph{et al}. firstly utilized the Boolean network-based method to investigate the transformation between these two types of FSRs in \cite{Lu2018AUTp393}, where an FSR is regarded as a Boolean networks (BNs). Although the equivalent conditions were relaxed slightly, the FSRs transformed from Fibonacci FSRs can not be guaranteed to be Galois FSRs.

In this paper, we further attempt to use the BN-based method to investigate this transformation. BNs are a kind of logical systems, thus they were first proposed to model gene regulatory networks \cite{KauffmanS1969Jotbp437}. Similar to FSRs, the state of each node in a BN is also binary: $1$ or $0$. Moreover, each node has its own update function, which is also a Boolean function. At each time instant, the state of each node is updated by the corresponding update function. Recently, to deal with this discrete-time and also discrete-state system, the semi-tensor product (STP) of matrices was proposed by Cheng \emph{et al.} \cite{ChengD2010pX}. STP of matrices breaks the rule of traditional matrices multiplication. That is amount to say that, under the framework of STP, matrix $A$ with dimension $n\times m$ can multiply with matrix $B$ with dimension $p\times q$, where $m\neq p$. By STP, a discrete-time logical system precesses its corresponding algebraic state space representation rather than directly convert into the ANF. Inspired by the convenience of STP, some remarkable results on BNs were obtained, including but not limited to, controllability \cite{ChengD2009Ap1659,Zhong2019AMCp51}, observability \cite{Zhang2014cccp6854}, stability and stabilization \cite{Huang2020TNNLSp,xu2019TCp,Meng2017TACp4222,Li2017SIAMp3437,Huang2020isp205}, output tracking \cite{Li2016isp1,Zhu2019isp96}, block decoupling \cite{Yu2019tacp3129} and optimal
control \cite{WuTACp262}.

Recently, BNs are used to model FSRs in \cite{Lu2017SCISp,Zhong2018ITITp6429}. From the analytical point of view, the ANF can not explicitly reflect the relations between update functions and state transition. Relative to the ANF, under the framework of STP, Zhong \emph{et al.} in \cite{Zhong2015JCSSp783} first revealed the relation between update functions and state transition for Fibonacci FSRs, and it is helpful to construct the equivalent FSRs. In this paper, main contributions are stressed into the following points:
\begin{itemize}
\item{} For a given Fibonacci FSR with $n$ stages, we first develop an approach to construct the equivalent Galois FSRs with $n$ stages. Actually, in many existing results including \cite{Dubrova2009ITITp5263,Dubrova2010ITITp2961,Dubrova2014p187,Lu2018AUTp393}, the stage number of these two equivalent FSRs were coincident. However, according to the definition of transformation, only the output sequences are necessary to be identical, the number of stages is not. Furthermore, it is well-known that the power consumption and latency will increase as the number of registers increases. A natural question is that whether there exists an equivalent Galois FSR with fewer stages than the given Fibonacci FSR but generating the same output sequences. We prove that for Fibonacci FSRs, there does not exist an equivalent Galois FSR with fewer stages, but the converse is not true.
\item{} According to the approach proposed in this paper, a total of $(2^{n-1})!^{2}-1$ equivalent Galois FSRs can be constructed. As shown in Table \ref{T16}, logical operators always take up memory and increase propagation time. Therefore, while designing FSRs, it is desirable for us to use as fewer logical operators as possible to achieve satisfactory aims. As usual, we only focus on the transformation, and ignore how to optimize the number of logical operators. In this paper, we develop an algorithm to select the equivalent Galois FSR with the minimal operators and stages.
\item{} For arbitrary given Fibonacci FSR with $n$ stages, there must exist equivalent Galois FSRs with $n$ stages. However for the transformation, from Galois FSRs with $n$ stages to Fibonacci FSRs, it is not necessary to guarantee the same number of stages. A criterion is derived to determine whether there exists an equivalent Fibonacci FSR with same stages. If true, we continue to construct the equivalent Fibonacci FSRs with stages fewer than $n$. Otherwise, we further verify whether there exists equivalent Fibonacci FSRs with the minimal number of stages greater than $n$.
\end{itemize}

This paper is organized as follows: Section \ref{s2} presents some notations and basic definitions. In Section \ref{s3}, the transformation problem between Fibonacci and Galois FSRs is investigated. Moreover, two algorithms are respectively designed to reduce the number of logical operators and stages as much as possible. The last section concludes this paper.

\section{Preliminaries}\label{s2}
First of all, some necessary notations are introduced to simplify the presentation of the main content.
\begin{itemize}
  \item $[a,b]:=\{a,a+1,\cdots,b\}$ with $a$ and $b$ being positive integers;
  \item $\mathscr{D}:=\{0, 1\}$;
  \item $\text{Col}_i(A)$ is the $i$-th column of matrix $A$;
  \item $\delta_{n}^{i}:=\text{Col}_i(I_n)$, where $I_n$ represents the identity matrix with dimension $n$;
  \item $\Delta_{n}:=\{\delta_{n}^{1}, \delta_{n}^{2},\cdots, \delta_{n}^{n}\}$;
  \item Matrix $A$ is called a logical matrix if $\text{Col}_i(A)\subseteq\Delta_{m}$, $i\in[1,n]$;
  \item $\mathscr{L}_{m\times n}$ represents the set of all $m\times n$ logical matrices;
  \item Logical matrix $[\delta_{n}^{i_{1}}~\delta_{n}^{i_{2}}\cdots \delta_{n}^{i_{m}}]\in \mathscr {L}_{m\times n}$ is abbreviated to $\delta_{n}[i_{1}~i_{2}~\cdots~i_{m}]$ for easy expression.
      \item Assume that $S=\{\delta_{n}^{i_{1}},~\delta_{n}^{i_{2}},\cdots,\delta_{n}^{i_{r}}\}$, let $[S]$ be the set of $\{i_{1},~i_{2},\cdots, i_{r}\}$.
\end{itemize}

In a Fibonacci FSR, the state value of register $x_{i}$ is moved to the next register $x_{i-1}, i\in [2, n]$, except for register $x_{1}$. It claims that the corresponding update function $f_{i}$ can be presented as $f_{j}=x_{j+1}, j\in [1, n-1]$. Particularly, the state of register $n$ at the next time instant depends on the states of certain registers taking from the set $\{x_{1}, x_{2},\cdots, x_{n}\}$. Besides, the update function $f_{n}$ is given in the form as $f_{n}= f(x_{1}, x_{2},\cdots, x_{n})$, which is generally called the feedback function of Fibonacci FSRs. However, for a Galois FSR, each register has its own feedback function. Therefore, the Fibonacci FSR with $n$ registers reads:
\begin{equation}\label{T1}
\left\{\begin{aligned}
x_{1}(t+1)=&x_{2}(t),\\
\vdots&\\
x_{n-1}(t+1)=&x_{n}(t),\\
x_{n}(t+1)=&f(x_{1}, x_{2},\cdots, x_{n}),\\
\end{aligned}\right.
\end{equation}
and the Galois FSR can be described as follows:
\begin{equation}\label{T2}
\left\{\begin{aligned}
z_{1}(t+1)=&f_{1}(z_{1}, z_{2},\cdots, z_{n}),\\
\vdots&\\
z_{n-1}(t+1)=&f_{n-1}(z_{1}, z_{2},\cdots, z_{n}),\\
z_{n}(t+1)=&f_{n}(z_{1}, z_{2},\cdots, z_{n}).\\
\end{aligned}\right.
\end{equation}
Thereinto, the state of the $i$-th register for Galois FSRs is presented as $z_{i}$, in order to avoid confusion with that of Fibonacci FSRs. Moreover, states $x_{1}(t)$ and $z_{1}(t)$ act as the outputs of Fibonacci FSRs and Galois FSRs at time instant $t$, respectively. Observing the implementing configurations of Fibonacci FSRs, it is noticed that the state transition is quite special. More precisely, for a given state $(x_{1}(t), x_{2}(t),\cdots, x_{n}(t))$, the next state of the whole register can be calculated as $(x_{2}(t), x_{3}(t), \cdots, f(x_{1}, x_{2},\cdots, x_{n}))$, which is called the successor of state $(x_{1}(t), x_{2}(t),\cdots, x_{n}(t))$.

It can be observed from Fig. \ref{Tf1} that these two FSRs have different structures. From the theoretical and practical points of views, they have own advantages and disadvantages. Thus, it is interesting and important for us to investigate the mutual transformation between these two kinds of FSRs.
\begin{defn}\cite{Dubrova2009ITITp5263}\label{T6}
Two FSRs are said to be equivalent if their sets of output sequences are equal.
\end{defn}

Afterwards, an FSR is regarded as a BN, then the STP method is applied to further investigate the transformation between FSRs. Before presenting the algebraic state space representation, some necessary acknowledge about STP and BNs are briefly introduced.
\begin{defn}\cite{ChengD2010pX}
The STP of two matrices $A\in M_{m\times n}$ and $B\in M_{p\times q}$ is defined as
\begin{equation}
A\ltimes B=\left(A\otimes I_{\frac{l}{n}}\right)\left(B\otimes I_{\frac{l}{p}}\right),
\end{equation}
where `$\otimes$' is the Kronecker product and $l=lcm(n, p)$ is the least common multiple of $n$ and $p$.
\end{defn}
\begin{lem}\cite{ChengD2010pX}\label{cg1}
Considering a logical function $f(x_{1}, x_{2},\cdots,x_{n}): \mathscr{D}^{n}\rightarrow\mathscr{D}$, there exists a unique matrix $M_{f}\in \mathscr{L}_{2\times 2^{n}}$, named as the structure matrix of $f$, such that
\begin{equation*}
f(x_{1}, x_{2},\cdots,x_{n})=M_{f}\ltimes_{i=1}^{n}x_{i},
\end{equation*}
where $\ltimes_{i=1}^{n}x_{i}=x_{1}\ltimes x_{2}\cdots\ltimes x_{n}\in \Delta_{2^n}$. Please refer to \cite{ChengD2010pX} for more detailed computation process.
\end{lem}

As mentioned in the Introduction, a BN is essentially a logical system, where the state of each node takes value from $\mathscr{D}$ and its update function is a logical function \cite{KauffmanS1969Jotbp437,ChengD2010pX}. As usual, to convert a BN system into the conventional discrete-time linear system, two equivalent vector forms of Boolean variables are defined: $1\sim \delta_{2}^{1}$ and $0\sim \delta_{2}^{2}$. Henceforth, we consider the corresponding canonical vectors of Boolean variables. By Lemma \ref{cg1}, the componentwise form of each subsystem can be obtained. For example, in Fibonacci FSR (\ref{T1}), assume that the structure matrix of feedback function $f$ is given as $M_{f}=\delta_{2}[i_{1}~i_{2}\cdots~i_{2^{n}}]$, it claims that
$x_{n}(t+1)=M_{f}x(t)$ with $x(t)=\ltimes_{i=1}^{n}x_{i}(t)\in \Delta_{2^{n}}$. Furthermore, by resorting to the STP method, the algebraic form of Fibonacci FSR (\ref{T1}) can be obtained as follows:
\begin{equation}\label{T5}
x(t+1)=L_{f}\ltimes x(t),
\end{equation}
where $L_{f}\in \mathscr{L}_{2^{n}\times 2^{n}}$ is called the transition matrix. The relation between $L_{f}$ and $M_{f}$ has been revealed in \cite{Zhong2015JCSSp783}, specifically,
$$L_{f}=\delta_{2^{n}}[q_{1}~q_{2}\cdots~q_{2^{n}}],$$
with
\begin{equation}\label{T8}
\left\{\begin{aligned}
&q_{j}=2j-2+i_{j},\\
&q_{2^{n-1}+j}=2j-2+i_{2^{n-1}+j},
\end{aligned}\right.~j\in [1, 2^{n-1}].
\end{equation}
Or alternatively, equation (\ref{T8}) can be presented in the form of
\begin{equation}\label{T10}
\left\{\begin{aligned}
&(L_{f})_{j}=2j-2+i_{j},\\
&(L_{f})_{2^{n-1}+j}=2j-2+i_{2^{n-1}+j}.
\end{aligned}\right.
\end{equation}
Therefore, once the structure matrix of the update function is known, the transition matrix can be calculated, and state transition can also be determined.

\section{Main Results}\label{s3}
In the following sequel, we will utilize the BN-based approach to further investigate the transformation between Fibonacci and Galois FSRs. Let $z(t)=\ltimes_{i=1}^{n}z_{i}(t)$ be the state of Galois FSR (\ref{T2}) at time instant $t$. Followed by Definition \ref{T6}, a coordinate transformation $z(t)=Tx(t)$ need to be found such that $z_{1}(t)=x_{1}(t)$. Or alternatively, for arbitrary initial state $x(0)$, we can find the corresponding initial state $z(0)$ such that the generating output sequences are equivalent.

\begin{defn}\label{T30}\cite{limniotis2007TITp4293}
For any given binary sequence $S$ assumed by
$S=(a_{1}, a_{2}, a_{3},\\\cdots)$, if there exist
$n_{1}> 0$ and $n_{2}> 0$ such that
$a_{i}=a_{i+n_{1}}$ for all $i\geq n_{2}$, then
sequence $S$ is called ultimately periodic.
The least integers $n_{1}$ and $n_{2}$ with this
property are called period and preperiod of
sequence $S$, respectively. Additionally, if
$n_{2}=1$, then sequence $S$ is called
periodic.
\end{defn}
Since FSRs are usually applied into digital circuits which only have two basic elements: $0$ and $1$, then the output sequences are composed of $0$ and $1$ rather than the corresponding canonical vectors. Therefore, if the output sequence is $1~ 0~0~1$, it implies that the corresponding state sequence of register $x_{1}$ is $\delta_{2}^{1}~\delta_{2}^{2}~ \delta_{2}^{2}~ \delta_{2}^{1}$. Because the order of elements has four cases: $1\rightarrow 1$, $1\rightarrow 0$, $0\rightarrow 1$ and $0\rightarrow 1$, we can define the following sets:

{\footnotesize\begin{equation*}
\begin{aligned}
&\Omega_{1}=\left\{\delta_{2^{n}}^{i}: i\in [1, 2^{n-1}]\right\};\\
&\Omega_{0}=\left\{\delta_{2^{n}}^{i}: i\in [2^{n-1}+1, 2^{n}]\right\};\\
&\mathcal {S}_{1\rightarrow 1}=\left\{\left(\delta_{2^{n}}^{i}, \delta_{2^{n}}^{j}\right): [L_{f}]_{ji}=1~\&~ \lceil\frac{i}{2^{n-1}}\rceil=\lceil\frac{j}{2^{n-1}}\rceil=1\right\};\\
&\mathcal {S}_{1\rightarrow 0}=\left\{\left(\delta_{2^{n}}^{i}, \delta_{2^{n}}^{j}\right): [L_{f}]_{ji}=1~\&~\lceil\frac{i}{2^{n-1}}\rceil=1~\&~ \lceil\frac{j}{2^{n-1}}\rceil=2\right\};\\
&\mathcal {S}_{0\rightarrow 1}=\left\{\left(\delta_{2^{n}}^{i}, \delta_{2^{n}}^{j}\right): [L_{f}]_{ji}=1~\&~\lceil\frac{i}{2^{n-1}}\rceil=2~\&~ \lceil\frac{j}{2^{n-1}}\rceil=1\right\};\\
&\mathcal {S}_{0\rightarrow 0}=\left\{\left(\delta_{2^{n}}^{i}, \delta_{2^{n}}^{j}\right): [L_{f}]_{ji}=1~\&~\lceil\frac{i}{2^{n-1}}\rceil=2~\&~ \lceil\frac{j}{2^{n-1}}\rceil=2\right\}.\\
\end{aligned}
\end{equation*}}

For any given $\alpha, \beta\in \mathscr{D}$, if there exists binary pair $\left(\delta_{2^{n}}^{i}, \delta_{2^{n}}^{j}\right)\in \mathcal {S}_{\alpha\rightarrow \beta}$, then it implies that the corresponding two adjacent outputs are $\alpha\rightarrow \beta$. Moreover, $\delta_{2^{n}}^{j}$ is called the successor of $\delta_{2^{n}}^{i}$, and $\delta_{2^{n}}^{i}$ is said to be the predecessor. It follows from equations (\ref{T8}) or (\ref{T10}) that the cardinal number of set $\mathcal {S}_{\alpha\rightarrow \beta}, \alpha, \beta\in \mathscr{D}$ is fixed, namely $2^{n-2}$.

\subsection{From Fibonacci FSRs to Galois FSRs}
In this subsection, we are in position to construct the equivalent Galois FSRs with the matching initial states for a given Fibonacci FSR with arbitrary initial states. Before giving the construct approach, a lemma is first presented.

\begin{lem}\label{T25}
If a Fibonacci FSR can be transformed into an equivalent Galois FSR, let the coordinate transformation be $z(t)=Tx(t)$, then matrix $T$ must be nonsingular.
\end{lem}
\begin{proof}
Assume that matrix $T$ is singular, that is, there exist $\tau_{1}\neq \tau_{2}$ and $\tau_{1}, \tau_{2}\in [1, 2^{n}]$ such that $\text{Col}_{\tau_{1}}(T)=\text{Col}_{\tau_{2}}(T)$. It amounts to say that two different states $\delta_{2^{n}}^{\tau_{1}}$ and $\delta_{2^{n}}^{\tau_{2}}$ of the Fibonacci FSR are assigned into the same state of the Galois FSR. Based on this observation, the Fibonacci FSR with initial states $\delta_{2^{n}}^{\tau_{1}}$ and $\delta_{2^{n}}^{\tau_{2}}$ can generate the same outputs sequences. Thereby, according to state transition property of Fibonacci FSRs, it implies that $\delta_{2^{n}}^{\tau_{1}}=\delta_{2^{n}}^{\tau_{2}}$, which is in contradiction with the assumption that $\tau_{1}\neq \tau_{2}$. Therefore, matrix $T$ is verified to be nonsingular.
\end{proof}

\begin{thm}\label{T9}
For any given $n$ stages' Fibonacci FSR, all equivalent Galois FSRs with $n$ stages can be constructed as follows:
\begin{itemize}
  \item For any $\alpha, \beta\in \mathscr{D}$ and $(\delta_{2^{n}}^{i}, \delta_{2^{n}}^{j})\in \mathcal {S}_{\alpha\rightarrow \beta}$, the unique corresponding binary array $(\delta_{2^{n}}^{i'}, \delta_{2^{n}}^{j'})$ are defined and denoted by $(\delta_{2^{n}}^{i}, \delta_{2^{n}}^{j})\sim (\delta_{2^{n}}^{i'}, \delta_{2^{n}}^{j'})$, where $\delta_{2^{n}}^{i'}\in \Omega_{\alpha}$ and $\delta_{2^{n}}^{j'}\in \Omega_{\beta}$;
  \item Matrix $T$ is nonsingular;
  \item $(\delta_{2^{n}}^{i'}, \delta_{2^{n}}^{j'})\neq (\delta_{2^{n}}^{i}, \delta_{2^{n}}^{j})$ holds for all binary arrays.
\end{itemize}
\end{thm}
\begin{proof}
For arbitrary given Fibonacci FSR, according to equation (\ref{T10}), it concludes that sets $\mathcal {S}_{\alpha\rightarrow \beta}, \alpha, \beta\in \mathscr {D}$ can be determined, and each set contains binary arrays which have the same order of elements. On the one hand, if the first and second conditions are satisfied, then we can determine the state trajectory of the equivalent FSR, which can guarantee the same output sequences. Thus, the transition matrix of the equivalent FSR, denoted by $L_{g}$, can be determined. In \cite{ChengD2010pX}, Cheng \emph{et al.} have showed the method to convert the transition matrix $L_{g}$ back to logical form (\ref{T2}). Thus, refer to \cite{ChengD2010pX}, the equivalent FSR can be constructed. On the other hand, we should guarantee that the equivalent FSR is Galois configuration. If two Fibonacci FSRs are equivalent, then they must be identical. Therefore, the third condition can ensure that the equivalent FSRs must be Galois FSRs. Moreover, by the above analysis, all the equivalent Galois FSRs must be satisfy the above three conditions.
\end{proof}

\begin{rmk}
Theorem \ref{T9} presents an efficient approach to construct all equivalent $n$ stages' Galois FSRs, while \cite{Lu2018AUTp393} only find a unique Galois. Moreover, the method in \cite{Lu2018AUTp393} can not guarantee that the constructed FSR is Galois configuration. Thus, Theorem \ref{T9} develops an improved transformation between Fibonacci and Galois FSRs.
\end{rmk}

According to Theorem \ref{T9}, for any given Fibonacci FSR with $n$ stages, there always exist at least one equivalent Galois FSR with $n$ stages. It is therefore natural to ask whether there exists equivalent Galois FSR which has fewer stages than the Fibonacci FSR while generating the same outputs sequences. According to Lemma \ref{T25}, the answer can be obtained immediately.
\begin{thm}\label{T26}
For any given Fibonacci FSR with $n$ stages, there does not exist any equivalent Galois FSR which has fewer stages than $n$.
\end{thm}
Actually if the equivalent Galois FSR has fewer stages, then the matrix $T$ must be singular, which is in contradiction with Lemma \ref{T25}.
\begin{rmk}\label{T14}
For a given Fibonacci FSR, more than one equivalent Galois FSR can be constructed based on Theorem \ref{T9}, because the update functions may be different. There exist a total of $(2^{n-1})!^{2}$ FSRs satisfying the first and second conditions. Therefore, for a given Fibonacci FSR, there are totally $(2^{n-1})!^{2}-1$ equivalent Galois FSRs.
\end{rmk}

\begin{rmk}\label{T18}
According to Theorem \ref{T9}, for an arbitrary given Fibonacci FSR, there must exist equivalent Galois FSRs. Compared with \cite{Dubrova2009ITITp5263}, the uniform form is unnecessary to be required in Theorem \ref{T9}. Moreover, the constructed FSRs must be Galois configuration. However, \cite{Lu2018AUTp393} focused on the output sequences, and failed to analyze the configurations of constructed FSRs.
\end{rmk}
\begin{example}\label{T11}
Consider the following Fibonacci FSR with $4$ stages:
\begin{equation}\label{T12}
\left\{\begin{aligned}
x_{1}(t+1)=&x_{2}(t),\\
x_{2}(t+1)=&x_{3}(t),\\
x_{3}(t+1)=&x_{4}(t),\\
x_{4}(t+1)=&(x_{1}(t)\wedge \neg x_{2}(t)\wedge \neg x_{3}(t)\wedge x_{4}(t))\\&\vee
(\neg x_{1}(t)\wedge (x_{2}(t)\wedge x_{3}(t))).\\
\end{aligned}\right.
\end{equation}
Followed by Lemma \ref{cg1}, the structure matrix of the feedback function is $M_{f}=\delta_{2}[2~2~2~2~2~2~1~2~1~1~1~1~1~1~2~2]$. Furthermore, by equation (\ref{T8}) or (\ref{T10}), we can calculate that $L_{f}=\delta_{16}[2~4~6~8~10~12~13~16~1~3~5~7~9~11~14~15]$. Moreover,
\begin{equation*}
\begin{aligned}
&\Omega_{1}=\left\{\delta_{16}^{i}:~i\in [1, 8]\right\};\\
&\Omega_{0}=\left\{\delta_{16}^{j}:~j\in [9, 16]\right\};\\
&\mathcal{S}_{1\rightarrow 1}=\left\{(\delta_{16}^{3}, \delta_{16}^{6}), (\delta_{16}^{1}, \delta_{16}^{2}), (\delta_{16}^{2}, \delta_{16}^{4}), (\delta_{16}^{4}, \delta_{16}^{8})\right\};\\
&\mathcal{S}_{1\rightarrow 0}=\left\{(\delta_{16}^{5}, \delta_{16}^{10}), (\delta_{16}^{6}, \delta_{16}^{12}), (\delta_{16}^{7}, \delta_{16}^{13}), (\delta_{16}^{8}, \delta_{16}^{16})\right\};\\
&\mathcal{S}_{0\rightarrow 1}=\left\{(\delta_{16}^{11}, \delta_{16}^{5}), (\delta_{16}^{10}, \delta_{16}^{3}), (\delta_{16}^{12}, \delta_{16}^{7}), (\delta_{16}^{9}, \delta_{16}^{1})\right\};\\
&\mathcal{S}_{0\rightarrow 0}=\left\{(\delta_{16}^{15}, \delta_{16}^{14}), (\delta_{16}^{14}, \delta_{16}^{11}), (\delta_{16}^{13}, \delta_{16}^{9}), (\delta_{16}^{16}, \delta_{16}^{15})\right\}.\\
\end{aligned}
\end{equation*}
According to Theorem \ref{T9}, for $\mathcal{S}_{1\rightarrow 1}$, let $(\delta_{16}^{3}, \delta_{16}^{6})\sim (\delta_{16}^{2}, \delta_{16}^{5})$, $(\delta_{16}^{1}, \delta_{16}^{2})\sim (\delta_{16}^{1}, \delta_{16}^{3})$, $(\delta_{16}^{2}, \delta_{16}^{4})\sim (\delta_{16}^{3}, \delta_{16}^{4})$ and $(\delta_{16}^{4}, \delta_{16}^{8})\sim (\delta_{16}^{4}, \delta_{16}^{8})$; For $\mathcal{S}_{1\rightarrow 0}$, let $(\delta_{16}^{5}, \delta_{16}^{10})\sim (\delta_{16}^{7}, \delta_{16}^{9})$, $(\delta_{16}^{6}, \delta_{16}^{12})\sim (\delta_{16}^{5}, \delta_{16}^{10})$, $(\delta_{16}^{7}, \delta_{16}^{13})\sim (\delta_{16}^{6}, \delta_{16}^{16})$ and $(\delta_{16}^{8}, \delta_{16}^{16})\sim (\delta_{16}^{8}, \delta_{16}^{13})$; For $\mathcal{S}_{0\rightarrow 1}$, let $(\delta_{16}^{11}, \delta_{16}^{5})\sim (\delta_{16}^{12}, \delta_{16}^{7})$, $(\delta_{16}^{10}, \delta_{16}^{3})\sim (\delta_{16}^{9}, \delta_{16}^{2})$, $(\delta_{16}^{12}, \delta_{16}^{7})\sim (\delta_{16}^{10}, \delta_{16}^{6})$ and $(\delta_{16}^{9}, \delta_{16}^{1})\sim (\delta_{16}^{14}, \delta_{16}^{1})$; For $\mathcal{S}_{0\rightarrow 0}$, let $(\delta_{16}^{15}, \delta_{16}^{14})\sim (\delta_{16}^{15}, \delta_{16}^{11})$, $(\delta_{16}^{14}, \delta_{16}^{11})\sim (\delta_{16}^{11}, \delta_{16}^{12})$, $(\delta_{16}^{13}, \delta_{16}^{9})\sim (\delta_{16}^{16}, \delta_{16}^{14})$ and $(\delta_{16}^{16}, \delta_{16}^{15})\sim (\delta_{16}^{13}, \delta_{16}^{15})$. Obviously, for all the above binary arrays, there is a binary array $(\delta_{16}^{1}, \delta_{16}^{3})\neq (\delta_{16}^{1}, \delta_{16}^{2})$. Therefore, the corresponding Galois FSR can be constructed, and the transition matrix can be computed as $L_{g}=\delta_{16}[3~5~4~8~10~16~9~13~2~6~12~7~15~1~11~14]$. The coordinate transformation $z(t)=Tx(t)$ is associated with $T=\delta_{16}[1~3~2~4~7~5~6\\8~14~9~12~10~16~11~15]$. Afterwards, refer to \cite{ChengD2010pX}, the corresponding logical form of Galois FSR can be obtained as follows:
\begin{equation}\label{T15}
\left\{\begin{aligned}
z_{1}(t+1)=&[z_{1}(t)\wedge z_{2}(t)]\vee [\neg z_{1}(t)\wedge [(z_{2}(t)\wedge (z_{3}(t)\\&\vee \neg z_{4}(t)))\vee \neg(z_{2}(t)\vee (z_{3}(t)\rightarrow z_{4}(t)))]],\\
z_{2}(t+1)=&[z_{1}(t)\wedge z_{4}(t)]\vee [\neg z_{1}(t)\wedge [(z_{2}(t)\wedge z_{4}(t))\\&\vee (\neg z_{2}\wedge (z_{3}(t)\oplus z_{4}(t)))]],\\
z_{3}(t+1)=&[z_{1}(t)\wedge [(z_{2}(t)\wedge (z_{3}(t)\oplus z_{4}(t)))\vee
(z_{2}(t)\\&\wedge (z_{3}(t)\rightarrow z_{4}(t)))]]\vee [\neg z_{1}(t)\wedge [(z_{2}(t)\\&\wedge z_{3}(t))\wedge \neg(z_{2}(t)\vee z_{4}(t))]],\\
z_{4}(t+1)=&[z_{1}(t)\wedge [(z_{2}(t)\wedge z_{3}(t))\vee \neg (z_{2}(t)\vee z_{3}(t))]]\\&\vee [\neg z_{1}(t)\wedge [(z_{2}(t)\wedge \neg(z_{3}(t)\vee z_{4}(t)))\vee \\&(\neg z_{2}(t)\wedge (z_{3}(t)\vee z_{4}(t)))]].
\end{aligned}\right.
\end{equation}

According to Theorem \ref{T9}, systems (\ref{T12}) and (\ref{T15}) are equivalent. As an illustration, system (\ref{T12}) with initial state $\delta_{16}^{1}$ generates the following output sequence with period $16$:
$$1~1~1~0~0~0~0~1~0~1~1~0~1~0~0~1\cdots.$$
Additionally, system (\ref{T14}) with initial state $\delta_{16}^{3}$ can generate the same output sequence with period $16$. Actually, according to Remark \ref{T14}, there exist a total of $(2^{3})!^{2}-1$ equivalent Galois FSRs. Due to the limitation of space, we only give a feasible equivalent Galois FSR.
\end{example}

It can be observed from Galois FSR (\ref{T15}) that every subsystem includes many logical gates, which decrease the speed of generating output sequences. Table \ref{T16} shows the space and time delays of some logical gates in a typical 90nm CMOS technology \cite{Dubrova2017espressop273}. Therefore, from the viewpoint of cost and efficiency, we should reduce the complexity of update functions when constructing the equivalent FSRs, that is, the number of logical gates. Afterwards, we propose the following algorithm to further reduce the number of logical gates and variables.
\begin{table}[H]
\caption{Parameters of gates for a typical 90nm CMOS technology.}\label{T16}
\centering
\begin{tabular}{|c|c|c|c|}
\hline
Gate & Area, $\mu m^{2}$ & Area, GE & Delay, ps\\
\hline
2-input NANS& 3.7& 1& 33\\
2-input NOR& 3.7& 1& 57\\
2-input AND& 5& 1.4& 87\\
2-input XOR& 10& 2.7& 115\\
\hline
\end{tabular}
\end{table}

\begin{algorithm}
\caption{To reduce the number of logical gates and variables.}
{\bf Input:} The set of all feasible transition matrices obtained from Theorem \ref{T9}: $K=\left\{L_g^1,L_g^2,\cdots,L_g^d\right\}$.

{\bf Output:} Galois FSR.

\begin{algorithmic}[1]
\Procedure{REDUCED TRANSITION MATRIX}{}

\State Set $i:=1$, $l:=n$ and $k=1$

\While{$i\leq d$}

  \State Assign $l_i=n$ and $j_i=1$
     \While{$j_i\leq n$}

     \If{$L_g^iW_{[2,2^{j-1}]}(M_n-I_2)=0$}

     \State Assign $l_i:=l_i-1$

     \State Increase $j_i$ by $1$

     \EndIf
     \EndWhile
     \If{$l_i\leq l$}
     \State Let $l:=l_i$
     \State Set $L_g^\ast:=L_g^i$
     \Else
     \State Break if
     \EndIf

  \State Increase $i$ by $1$
\EndWhile

\While{$k\leq n$}
\State Define $S_k^n:=S_1^nW_{[2^{k-1},2]}$
\State Compute $M_k^\ast:=S_k^nL_g^\ast$
\State Set $j_k:=1$ and $\Lambda_k:=\{1,2,\cdots,n\}$
\While{$j_k \leq n$}
\If{$M_k^\ast W_{[2,2^{j_k-1}]}(M_n-I_2)=0$}
\State Assign $\Lambda_k:=\Lambda_k\backslash \{j_k\}$
\State Calculate $M_k:=M_kW_{[2,2^{j-1}]}\delta_2^1$
\EndIf
\State Denote $x_k(t+1)=M_k\ltimes_{j_r\in\Lambda_k}x_{j_r}(t)$, $r=1,2,\cdots,\mid \Lambda_k \mid$
\State Solve the logical function $f_k$
\EndWhile
\EndWhile

\Return Galois FSR (\ref{T2})
\EndProcedure
\end{algorithmic}\label{T17}
\end{algorithm}
\begin{rmk}
In Algorithm \ref{T17}, the formula in line 6 is utilized to judge whether the $i$-th logical system is dependent of variable $x_j$. Moreover, in line 20-21, $S_k^n$ is calculated to obtain the structure matrix of $f_k$. In the manner as that in line 6, the sub-procedure developed in line 24-26 is devoted to remove the independent variables.
\end{rmk}

\subsection{From Galois FSRs to Fibonacci FSRs }
It can be learned from Remark \ref{T14} that for any Fibonacci FSR, there must exist some equivalent Galois FSRs. Inversely, for any Galois FSR, whether there exist equivalent Fibonacci FSRs. Actually, there dose not exist a definitely answer, which depends on the characters of the given Galois FSR. Assume that there are two attractors given by a fixed point and a cyclic attractor with period $r$, that are $\delta_{2^{n}}^{\alpha_{0}}$ and $\delta_{2^{n}}^{\alpha_{1}}\rightarrow \delta_{2^{n}}^{\alpha_{2}}\rightarrow\cdots\rightarrow \delta_{2^{n}}^{\alpha_{r}}$ with $r\geqslant n$. Then, by resorting to the depth-first search algorithm, assume that there are two trajectories which traverse all the states, to attractors, given by
\begin{equation*}
\begin{aligned}
&(1)~\delta_{2^{n}}^{\beta_{1}}\rightarrow \delta_{2^{n}}^{\beta_{2}}\rightarrow\cdots
\rightarrow\delta_{2^{n}}^{\beta_{s_{1}}}\rightarrow
\delta_{2^{n}}^{\alpha_{0}}\rightarrow
\delta_{2^{n}}^{\alpha_{0}}\rightarrow\cdots;\\
&(2)~\delta_{2^{n}}^{\gamma_{1}}\rightarrow \delta_{2^{n}}^{\gamma_{2}}\rightarrow\cdots
\rightarrow\delta_{2^{n}}^{\gamma_{s_{2}}}\rightarrow
\delta_{2^{n}}^{\alpha_{1}}\rightarrow\cdots\rightarrow \delta_{2^{n}}^{\alpha_{r}}\rightarrow \delta_{2^{n}}^{\alpha_{1}}\rightarrow\cdots.
\end{aligned}
\end{equation*}
Moreover, the corresponding output sequences are assumed to be
\begin{equation*}
\begin{aligned}
&(1)~a_{\beta_{1}}~a_{\beta_{2}}~a_{\beta_{3}}\cdots~a_{\beta_{n}}
\cdots~a_{\beta_{s_{1}}}~a_{\omega_{0}}~a_{\omega_{0}}\cdots;\\
&(2)~b_{\gamma_{1}}~b_{\gamma_{2}}~b_{\gamma_{3}}\cdots
~b_{\gamma_{s_{2}}}~b_{\omega_{1}}\cdots~b_{\omega_{r}}~b_{\omega_{1}}\cdots.\\
\end{aligned}
\end{equation*}
Here, $a_{\omega_{0}}$ and $b_{\omega_{i}}, i\in [1, r]$ are the corresponding outputs of states $\delta_{2^{n}}^{\alpha_{0}}$ and $\delta_{2^{n}}^{\alpha_{i}}$, respectively. Without loss of generality, we also assume that output sequence $(2)$ is ultimately periodic and the period equals to $r$.

If the equivalent Fibonacci FSR with $n$ stages can be constructed, followed by the property of states transition of Fibonacci FSRs, we have the following state trajectories£º
$(a_{\beta_{1}}, a_{\beta_{2}},\cdots, a_{\beta_{n}})\rightarrow (a_{\beta_{2}}, a_{\beta_{3}},\cdots, a_{\beta_{n+1}})\rightarrow\cdots\rightarrow(a_{\omega_{0}}, a_{\omega_{0}},\cdots, a_{\omega_{0}})$, and $(b_{\gamma_{1}}, b_{\gamma_{2}},\cdots, b_{\gamma_{n}})\rightarrow (b_{\gamma_{2}}, b_{\gamma_{3}},\cdots, b_{\gamma_{n+1}})\rightarrow\cdots\rightarrow(b_{\omega_{r}}, b_{\omega_{1}},\cdots,\\b_{\omega_{r-n+1}})$.

The derived state transition digraph of Galois FSRs, denoted by $G_{n}$, contains nodes and edges, where nodes represent the states of FSRs with $n$ stages. If state $(a_{1}, a_{2},\cdots, a_{n})$ is changed into $(a'_{1}, a'_{2},\cdots, a'_{n})$ via update functions, then there exists a directed edge from the node representing $(a_{1}, a_{2},\cdots, a_{n})$ to the node representing $(a'_{1}, a'_{2},\cdots, a'_{n})$. Additionally, let the output-degree of node $(a_{1}, a_{2},\cdots, a_{n})$ be the number of edges which leave from the node.
\begin{thm}\label{T19}
Consider Galois FSR (\ref{T2}) with outputs sequences $(1)$ and $(2)$, the equivalent Fibonacci FSR can be constructed, if and only if, the output degree of every node equals to $1$ in  derived state transition digraph $G_{n}$.
\end{thm}
\begin{proof}
State transition trajectories in digraph $\mathcal {G}_{n}$ are $(a_{\beta_{1}}, a_{\beta_{2}},\cdots, a_{\beta_{n}})\rightarrow (a_{\beta_{2}}, a_{\beta_{3}},\cdots, a_{\beta_{n+1}})\rightarrow\cdots\rightarrow(a_{\omega_{0}}, a_{\omega_{0}},\cdots, a_{\omega_{0}})\rightarrow(a_{\omega_{0}}, a_{\omega_{0}},\cdots, a_{\omega_{0}})$, and $(b_{\gamma_{1}}, b_{\gamma_{2}},\cdots,\\b_{\gamma_{n}})\rightarrow (b_{\gamma_{2}}, b_{\gamma_{3}},\cdots, b_{\gamma_{n+1}})\rightarrow\cdots\rightarrow(b_{\omega_{r}}, b_{\omega_{1}},\cdots, b_{\omega_{r-n+1}})\rightarrow(b_{\omega_{1}},b_{\omega_{2}},\cdots, b_{\omega_{n}})$. Assume that the corresponding canonical vector forms are $\delta_{2^{n}}^{\eta_{\beta_{1}}}\rightarrow \delta_{2^{n}}^{\eta_{\beta_{2}}}\rightarrow\cdots\rightarrow \delta_{2^{n}}^{\eta_{\omega_{0}}}$, and $\delta_{2^{n}}^{\kappa_{\gamma_{1}}}\rightarrow \delta_{2^{n}}^{\kappa_{\gamma_{2}}}\rightarrow\cdots \delta_{2^{n}}^{\kappa_{\omega_{r}}}\rightarrow\delta_{2^{n}}^{\kappa_{\omega_{1}}}$. Therefore, the transition matrix $L_{f}$ must satisfy
\begin{equation*}
\left\{\begin{aligned}
&\text{Col}_{\eta_{\beta_{1}}}(L_{f})=\delta_{2^{n}}^{\eta_{\beta_{2}}},\cdots, \text{Col}_{\eta_{\omega_{0}}}(L_{f})=\delta_{2^{n}}^{\eta_{\omega_{0}}};\\
&\text{Col}_{\kappa_{\gamma_{1}}}(L_{f})=\delta_{2^{n}}^{\kappa_{\gamma_{2}}},\cdots, \text{Col}_{\kappa_{\omega_{r}}}(L_{f})=\delta_{2^{n}}^{\kappa_{\omega_{1}}}.\\
\end{aligned}\right.
\end{equation*}
Moreover, the output degree of every node equals to $1$. If there exists one node, whose output degree equals to $2$, then it implies that there exists one column of $L_{f}$ equaling two values. Obviously, there does not exist such $L_{f}$. Therefore, some columns of $L_{f}$ can be determined and satisfy equation (\ref{T10}). Additionally, the rest columns satisfy equation (\ref{T10}). Thus, the constructed Fibonacci FSR with $L_{f}$ is equivalent to Galois FSR (\ref{T2}).
\end{proof}
\begin{lem}\label{T33}
For any given binary sequence, there exists a Galois FSR at least that can generate the binary sequence, but it is not true for Fibonacci FSRs.
\end{lem}
\begin{proof}
Assume that any given binary sequence denoted by $S'$, which is ultimately periodic, is $a_{1}~a_{2}~a_{3}\cdots~$ with its period and preperiod being $n_{1}$ and $n_{2}$, then the Galois FSRs with $n,~n\geq \text{log}_2(n_{1}+n_{2})$, stages can be constructed as: $Col_{i}(L_{g})=\delta_{2^{n}}^{j}$ where $i\in [\Omega_{a_{i}}]$ and $j\in [\Omega_{a_{i+1}}]$. While, the sequence $S'$ can not be guaranteed satisfy Theorem \ref{T19}, then there possibly dose not exist one Fibonacci FSR generating sequence $S'$, which completes the proof.
\end{proof}
It follows from Theorem \ref{T26} that there does not exist any equivalent Galois FSR with fewer stages than $n$ for any given Fibonacci FSR with $n$ stages. While for one given Galois FSR with $n$ stages, there may exist an equivalent Fibonacci FSR whose the number of stages is less than $n$. The following example shows the existence of this kind of Fibonacci FSR.
\begin{example}
Consider the following Galois FSR with $3$ stages:
\begin{equation}\label{T20}
\left\{\begin{aligned}
z_{1}(t+1)=&z_{1}(t)\vee \neg z_{2}(t),\\
z_{2}(t+1)=&(z_{1}(t)\wedge \neg z_{2}(t)\wedge z_{3}(t))\vee (\neg z_{1}(t)\wedge z_{2}(t)),\\
z_{3}(t+1)=&z_{1}(t)\wedge (z_{2}(t)\leftrightarrow z_{3}(t)).\\
\end{aligned}\right.
\end{equation}
Then one has that the transition matrix $L_{g}=\delta_{8}[3~4~2~3~6~6~4~4]$. By the depth-first search algorithm, four states trajectories containing all the states can be found, as: $\textcircled{1}~ \delta_{8}^{1}\rightarrow \delta_{8}^{3}\rightarrow \delta_{8}^{2}\rightarrow \delta_{8}^{4}\rightarrow \delta_{8}^{3}\rightarrow \delta_{8}^{2}\rightarrow\cdots$; $\textcircled{2}~ \delta_{8}^{5}\rightarrow \delta_{8}^{6}\rightarrow \delta_{8}^{6}\rightarrow\cdots$; $\textcircled{3}~ \delta_{8}^{8}\rightarrow \delta_{8}^{4}\rightarrow \delta_{8}^{3}\rightarrow \delta_{8}^{2}\rightarrow \delta_{8}^{4}\rightarrow\cdots$; $\textcircled{4}~ \delta_{8}^{7}\rightarrow \delta_{8}^{4}\rightarrow \delta_{8}^{3}\rightarrow \delta_{8}^{2}\rightarrow \delta_{8}^{4}\rightarrow\cdots$. Then the corresponding outputs sequences are: $\textcircled{1}~1~1~1~1~1~1\cdots$; $\textcircled{2}~ 0~0~0\cdots$; $\textcircled{3}~0~1~1~1~1\cdots$; $\textcircled{4}~ 0~1~1~1~1\cdots$. Let coordinate transformation be $x(t)=\delta_{4}[1~1~1~1~4~4~3~3]z(t)$ with $z(t)\in \Delta_{2^{3}}$ and $x(t)\in \Delta_{2^{2}}$ being the states of Galois FSR (\ref{T20}) and the equivalent Fibonacci FSR, respectively. Then we can construct two Fibonacci FSRs with $2$ stages which are equivalent to Galois FSR (\ref{T20}), that are
\begin{equation}\label{T21}
\left\{\begin{aligned}
x_{1}(t+1)=&x_{2}(t)\\
x_{2}(t+1)=&x_{1}\vee x_{2}(t),\\
\end{aligned}\right.
\end{equation}
and
\begin{equation}\label{T22}
\left\{\begin{aligned}
x_{1}(t+1)=&x_{2}(t)\\
x_{2}(t+1)=&x_{2}(t).\\
\end{aligned}\right.
\end{equation}
Take Fibonacci FSR (\ref{T21}) for example, when initial states are $x(0)=\delta_{4}^{1}, \delta_{4}^{4}, \delta_{4}^{3}, \delta_{4}^{3}$, respectively, Fibonacci FSR (\ref{T21}) can generate the same outputs sequences corresponding $\textcircled{1}-\textcircled{4}$.
\end{example}
It follows from the above example that there may exists an equivalent Fibonacci FSR which has fewer stages than that of the counterpart. Now, we prove the conclusion from the viewpoint of theoretical analysis.
\begin{thm}\label{T34}
For any given Galois FSR with $n$ stages, there possibly exists an equivalent Fibonacci FSR where the number of stages is less than $n$.
\end{thm}
\begin{proof}
To show the conclusion hold, we only construct one binary sequence such that the output degree of every node equals to $1$ in $G_{i}$ with $i< n$. Obviously, there must exists that binary sequence. The Lemma \ref{T33} suffices to show that there exist Galois FSRs with $n$ stages that can generate the binary sequence. Thus, for Galois FSR with $n$ stages, there potentially exists one equivalent Fibonacci FSR where the number of stages is less than $n$.
\end{proof}

\begin{rmk}\label{T27}
It can be learned from \cite{Dubrova2009ITITp5263} that Galois FSR (\ref{T20}) fails to satisfy the uniform conditions, but there still exist the equivalent Fibonacci FSRs. Therefore, compared with \cite{Dubrova2009ITITp5263,Dubrova2010ITITp2961,Liu2013p335}, our method is less conservative to investigate the transformation between these two types of FSRs.
\end{rmk}
Now, for a given Galois FSR with $n$ stages, an algorithm is proposed to find equivalent Fibonacci FSRs with minimal stages.
\begin{lem}
If there does not exist any equivalent Fibonacci FSR with $k$ stages for a given Galois FSR with $n$ stages, then it cannot have an equivalent Fibonacci FSR, the number of whose stages is less than $k$.
\end{lem}
\begin{algorithm}[h!]
\caption{To construct the equivalent Fibonacci FSRs with the minimal stages.}
{\bf Input:} Output sequences (1) and (2).

{\bf Output:} The Fibonacci FSR with the minimal stages.

\begin{algorithmic}[1]
\Procedure{REDUCED STAGES}{}

\State Set $l:=\lceil\text{log}_2(r)\rceil$

\State Construct the derived state transition digraph $G_l$

\If{the output degree of per node is $1$}

\State Construct logical matrix $L_f$ by $G_{l}$

\Return Matrix $L_f$ and the minimal number of stages is $l$

\Else

\State Assign $l:=l+1$

\EndIf
\EndProcedure
\end{algorithmic}\label{T23}
\end{algorithm}
\begin{rmk}\label{T28}
For one Fibonacci FSR with $n$ stages, the period of its state trajectory is the same as that of its corresponding output sequence, and the maximum period of output sequences equals to $2^{n}$. Then, for constructed equivalent Fibonacci FSR, its maximum period should be greater than that of Galois FSR (\ref{T2}). Therefore, in Algorithm \ref{T23}, the minimal stages of the equivalent Fibonacci FSR should be greater than $\text{log}_2(r)$.
\end{rmk}
\begin{example}\label{T24}
Let' s consider the following $3-$stage Galois FSR:
{\footnotesize\begin{equation}\label{T29}
\left\{\begin{aligned}
z_{1}(t+1)=&[z_{1}(t)\wedge \neg(z_{2}(t)\rightarrow z_{3}(t))]\vee (\neg z_{1}(t)\wedge z_{2}(t)),\\
z_{2}(t+1)=&[z_{1}(t)\wedge (z_{2}(t)\leftrightarrow z_{3}(t))]\vee \neg [z_{1}(t)\vee (z_{2}(t)\rightarrow z_{3}(t))],\\
z_{3}(t+1)=&[z_{1}(t)\wedge(z_{2}(t)\vee z_{3}(t))]\vee
\neg(z_{1}(t)\vee z_{3}(t)),\\
\end{aligned}\right.
\end{equation}}
with transition matrix $L_{g}=\delta_{8}[5~3~7~6~4~1~8~7]$. Obviously, system (\ref{T29}) has two cyclic attractors, that are: $\delta_{8}^{7}\rightarrow \delta_{8}^{8}$ and $\delta_{8}^{1}\rightarrow \delta_{8}^{5}\rightarrow\delta_{8}^{4}\rightarrow \delta_{8}^{6}$, then all the state trajectories can be obtained as:
\begin{equation*}
\begin{aligned}
&(i)~\delta_{8}^{2}\rightarrow \delta_{8}^{3}
\rightarrow\delta_{8}^{7}\rightarrow
\delta_{8}^{8}\rightarrow
\delta_{8}^{7}\rightarrow\cdots;\\
&(ii)~\delta_{8}^{1}\rightarrow \delta_{8}^{5}
\rightarrow\delta_{8}^{4}\rightarrow
\delta_{8}^{6}\rightarrow \delta_{8}^{1}\rightarrow\cdots.
\end{aligned}
\end{equation*}
Thus, the corresponding output sequences are:
\begin{equation*}
\begin{aligned}
&(i)~1~1~0~0~0\cdots;\\
&(ii)~1~0~1~0~1\cdots.\\
\end{aligned}
\end{equation*}
Clearly, output sequence $(i)$ is ultimately periodic and its period is $1$, but output sequence $(ii)$ is periodic and its period equals to $2$. According to Algorithm \ref{T23}, one has $l=\lceil\text{log}_2(2)\rceil=1$. While, $G_1$ and $G_2$ fail to satisfy Theorem \ref{T19}, then we consider the case $l=3$. It can be observed from Fig. \ref{Tf2} that equivalent Fibonacci FSRs with $3$ stages can be constructed.
\begin{figure}[H]
\centering
\includegraphics[width=0.3\textwidth]{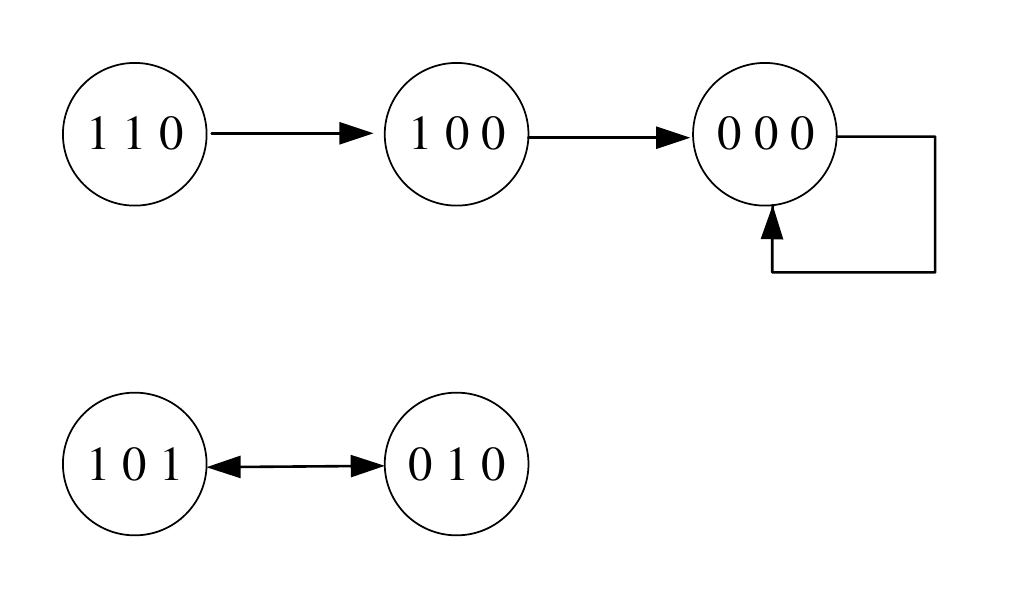}
\caption{The derived state transition digraph $G_{3}$ of Galois FSR (\ref{T29}).}
\label{Tf2}
\end{figure}
Moreover, one can conclude that the transition matrix $L_{f}$ is in the form of
\begin{equation}\label{T31}
L_{f}=\delta_{8}[\ast~4~6~8~\ast~3~\ast~8]s
\end{equation}
and coordinate transformation is $x(t)=T'z(t)$ with
$$T'=\delta_{8}[3~2~4~3~6~6~8~8].$$
Clearly, output sequences $(i)$ and $(ii)$ can be generated by Fibonacci FSRs with transition matrix being (\ref{T31}) when initial states equal to $\delta_{8}^{2}$ and $\delta_{8}^{3}$, respectively. It should be pointed out that in (\ref{T31}), the columns (represented by $\ast$) of $L_{f}$ only need to satisfy equation (\ref{T10}), which implies that there exist eight equivalent Fibonacci FSRs with $3$ stages. For example, Let $L_{f}=\delta_{8}[1~4~6~8~2~3~5~8]$, then the corresponding Fibonacci FSR is:
\begin{equation}\label{T32}
\left\{\begin{aligned}
x_{1}(t+1)=&x_{2}(t),\\
x_{2}(t+1)=&x_{3}(t),\\
x_{3}(t+1)=&[x_{1}(t)\wedge x_{2}(t)\wedge x_{3}(t)]\\&\vee
[\neg x_{1}(t)\wedge (x_{2}(t)\oplus x_{3}(t))].\\
\end{aligned}\right.
\end{equation}
\end{example}
\section{Conclusions}\label{s4}
In this paper, we investigated the transformation between Fibonacci FSRs and Galois FSRs. We regarded these two FSRs as BNs, then under the framework of STP, the corresponding algebraic expressions can be obtained. One approach was proposed to construct equivalent Galois FSRs, and the number of logical operators were optimized. We also provided a criterion to see whether there exists one equivalent Fibonacci FSR. Moreover, one algorithm was given to find the equivalent Fibonacci FSRs with fewer stages. Additionally, it is interesting to find out that for any Fibonacci FSRs, there does not exist any equivalent Galois FSRs with fewer stages.

\section*{References}

\begin{thebibliography}{10}

\bibitem{Jiang2018TITp3944}
Y.~Jiang and D.~Lin.
\newblock Lower and upper bounds on the density of irreducible {NFSRs}.
\newblock {\em IEEE Transactions on Information Theory}, 64(5):3944--3952,
  2018.

\bibitem{Dubrova2009ITITp5263}
E.~Dubrova.
\newblock A transformation from the {Fibonacci} to the {Galois} {NLFSRs}.
\newblock {\em IEEE Transactions on Information Theory}, 55(11):5263--5271,
  2009.

\bibitem{Zhang2015ITITp645}
J.~Zhang, W.~Qi, T.~Tian, and Z.~Wang.
\newblock Further results on the decomposition of an {NFSR} into the cascade
  connection of an {NFSR} into an {LFSR}.
\newblock {\em IEEE Transactions on Information Theory}, 61(1):645--654, 2015.

\bibitem{Zadeh2014IETp188}
A.~Zadeh and H.~Heys.
\newblock Simple power analysis applied to nonlinear feedback shift registers.
\newblock {\em IET Information Security}, 8(3):188--198, 2014.

\bibitem{Dubrova2010ITITp2961}
E.~Dubrova.
\newblock Finding matching initial states for equivalent {NLFSRs} in the
  {Fibonacci} and the {Galois} configurations.
\newblock {\em IEEE Transactions on Information Theory}, 56(6):2961--2966,
  2010.

\bibitem{Liu2013p335}
Z.~Liu.
\newblock The transformation from the {G}alois {NLFSR} to the {F}ibonacci
  configuration.
\newblock {\em 2013 Fourth International Conference on Emerging Intelligent
  Data and Web Technologies}, 335--339, 2013.

\bibitem{Dubrova2014p187}
E.~Dubrova.
\newblock An equivalence-preserving transformation of shift registers.
\newblock {\em International Conference on Sequences and Their Applications},
  187--199, 2014.

\bibitem{Lu2018AUTp393}
J.~Lu, M.~Li, T.~Huang, Y.~Liu, and J.~Cao.
\newblock The transformation between the {Galois} {NLFSRs} and the {Fibonacci}
  {NLFSRs} via semi-tensor product of matrices.
\newblock {\em Automatica}, 96:393--397, 2018.

\bibitem{KauffmanS1969Jotbp437}
S.~Kauffman.
\newblock Metabolic stability and epigenesis in randomly constructed genetic
  nets.
\newblock {\em Journal of Theoretical Biology}, 22(3):437--467, 1969.

\bibitem{ChengD2010pX}
D.~Cheng, H.~Qi, and Z.~Li.
\newblock Analysis and control of {Boolean} networks: A semi-tensor product
  approach.
\newblock {\em Springer Science \& Business Media}, 2010.

\bibitem{ChengD2009Ap1659}
D.~Cheng and H.~Qi.
\newblock Controllability and observability of {Boolean} control networks.
\newblock {\em Automatica}, 45(7):1659--1667, 2009.

\bibitem{Zhong2019AMCp51}
J.~Zhong, Y.~Liu, K.~Kou, L.~Sun, and J.~Cao.
\newblock On the ensemble controllability of {Boolean} control networks using
  {STP} method.
\newblock {\em Applied Mathematics and Computation}, 358:51--62, 2019.

\bibitem{Zhang2014cccp6854}
K.~Zhang and L.~Zhang.
\newblock Observability of {Boolean} control networks: A unified approach based
  on the theories of finite automata and formal languages.
\newblock {\em Proceedings of the 33rd Chinese Control Conference}, pages
  6854--6861, 2014.

\bibitem{Huang2020TNNLSp}
C.~Huang, J.~Lu, G.~Zhai, J.~Cao, G.~Lu, and M.~Perc.
\newblock Stability and stabilization in probability of probabilistic {Boolean}
  networks.
\newblock {\em IEEE Transactions on Neural Networks and Learning Systems}, Doi:
  10.1109/TNNLS.2020.2978345.

\bibitem{xu2019TCp}
M.~Xu, Y.~Liu, J.~Lou, Z.~Wu, and J.~Zhong.
\newblock Set stabilization of probabilistic {Boolean} control networks: A
  sampled-data control approach.
\newblock {\em IEEE Transactions on Cybernetics}, Doi:
  10.1109/TCYB.2019.2940654.

\bibitem{Meng2017TACp4222}
M.~Meng, L.~Liu, and G.~Feng.
\newblock Stability and $l_{1}$ gain analysis of {Boolean} networks with
  {Markovian} jump parameters.
\newblock {\em IEEE Transactions on Automatic Control}, 62(8):4222--4228, 2017.

\bibitem{Li2017SIAMp3437}
H.~Li and Y.~Wang.
\newblock Lyapunov-based stability and construction of {Lyapunov} functions for
  {Boolean} networks.
\newblock {\em SIAM Journal on Control and Optimization}, 55(6):3437--3457,
  2017.

\bibitem{Huang2020isp205}
C.~Huang, J.~Lu, D.W.C. Ho, G.~Zhai, and J.~Cao.
\newblock Stabilization of probabilistic {Boolean} networks via pinning control
  strategy.
\newblock {\em Information Sciences}, 510.

\bibitem{Li2016isp1}
H.~Li, Y.~Wang, and P.~Guo.
\newblock State feedback based output tracking control of probabilistic
  {Boolean} networks.
\newblock {\em Information Sciences}, 349.

\bibitem{Zhu2019isp96}
S.~Zhu, J.~Lu, Y.~Liu, T.~Huang, and J.~Kurths.
\newblock Output tracking of probabilistic {Boolean} networks by output
  feedback control.
\newblock {\em Information Sciences}, 483.

\bibitem{Yu2019tacp3129}
Y.~Yu, J.~Feng, J.~Pan, and D.~Cheng.
\newblock Block decoupling of {Boolean} control networks.
\newblock {\em IEEE Transactions on Automatic Control}, 64(8):3129--3140, 2019.

\bibitem{WuTACp262}
Y.~Wu and T.~Shen.
\newblock A finite convergence criterion for the discounted optimal control of
  stochastic logical networks.
\newblock {\em IEEE Transactions on Automatic Control}, 63(1):262--268, 2018.

\bibitem{Lu2017SCISp}
J.~Lu, M.~Li, Y.~Liu, D.W.C. Ho, and J.~Kurths.
\newblock Nonsingularity of grain-like cascade {FSR}s via semi-tensor product.
\newblock {\em Science China Information Sciences}, 61(1):010204, 2018.

\bibitem{Zhong2018ITITp6429}
J.~Zhong and D.~Lin.
\newblock On minimum period of nonlinear feedback shift registers in grain-like
  structure.
\newblock {\em IEEE Transactions on Information Theory}, 64(9):6429--6442,
  2018.

\bibitem{Zhong2015JCSSp783}
J.~Zhong and D.~Lin.
\newblock A new linearization method for nonlinear feedback shift registers.
\newblock {\em Journal of Computer and System Sciences}, 81(4).

\bibitem{limniotis2007TITp4293}
K.~Limniotis, N.~Kolokotronis, and N.~Kalouptsidis.
\newblock On the nonlinear complexity and lempel--ziv complexity of finite
  length sequences.
\newblock {\em IEEE Transactions on Information Theory}, 53(11):4293--4302,
  2007.

\bibitem{Dubrova2017espressop273}
E.~Dubrova and M.~Hell.
\newblock Espresso: A stream cipher for 5g wireless communication systems.
\newblock {\em Cryptography and Communications}, 9(2):273--289, 2017.

\end{thebibliography}

\end{document}